\documentclass[format=acmsmall, review=false]{acmart}

\usepackage{acm-ec-18-copy}

\usepackage{booktabs} 

			\newcommand{\pref}{\ensuremath{\succsim}}
			\newcommand{\spref}{\ensuremath{\succ}}
			
				\newcommand{\ceil}[1]{\lceil #1 \rceil }
				
				\newtheorem{remark}{Remark}
				\newtheorem{assumption}{Assumption}

\usepackage[noend]{algorithmic}

\usepackage[ruled]{algorithm2e} 

\SetAlFnt{\small}
\SetAlCapFnt{\small}
\SetAlCapNameFnt{\small}
\SetAlCapHSkip{0pt}
\IncMargin{-\parindent}

\begin{document}
\title[Representing the Insincere: Strategically Robust Proportional Representation]{Representing the Insincere: Strategically Robust Proportional Representation}  
\author{Barton E. Lee}

\begin{abstract}
Proportional representation (PR) is a fundamental principle of many democracies world-wide which employ PR-based voting rules to elect their representatives. The normative properties of these voting rules however, are often only understood in the context of sincere voting. 

In this paper we consider PR in the presence of strategic voters. We construct a voting rule such that for every preference profile there exists at least one costly voting equilibrium satisfying PR  with respect to voters' private and unrevealed preferences -- such a voting rule is said to be \emph{strategically robust}. In contrast, a commonly applied voting rule is shown not be strategically robust. Furthermore, we prove a limit on `how strategically robust' a PR-based voting rule can be; we show that there is no PR-based voting rule which ensures that every equilibrium satisfies PR. Collectively, our results highlight the possibility and limit of achieving PR in the presence of strategic voters and a positive role for mechanisms, such as pre-election polls, which coordinate voter behaviour towards equilibria which satisfy PR.
\end{abstract}


\maketitle


\section{Introduction}



\begin{quote}
The first principle of democracy [is] representation in proportion to numbers -- John Stuart Mill in \emph{Considerations on Representative Government} (1861).
\end{quote}

Multi-seat (or multi-winner) election outcomes which proportionally represent the diverse preferences of voters, are said to provide \emph{proportional representation} (PR). Informally speaking, PR captures the idea that if n\% of voters support a certain political party then roughly n\% of the seats should be allocated to this party. In democratic elections, PR has long been considered a major desideratum~\cite{Droo81a,Dumm84a,TiRi00a}. As per the political philosopher John Stuart Mill, PR is the ``first principle of democracy"~\cite{Mill61}. This view is also shared by many others, including~\cite{Wood97a} who commented that PR is ``a sine qua non for a fair election rule". In practice, PR is a fundamental goal of national elections world-wide including Australia, India, Ireland and Pakistan;  all of which elect parliamentary representatives via a voting rule understood to satisfy PR. Furthermore, countries which do not pursue PR in their national elections are often subject to vigorous debate from prominent electoral reform movements such as the Electoral Reform Society (UK), Fair Vote Canada, and FairVote (US). These movements have led to referenda over electoral reform in the UK (2011), Ontario (2007), and British Columbia (2005). In 2018, British Columbia is expected to hold another referendum for the adoption of a PR electoral system.

The family of voting systems used around the world to achieve PR, which we call PR-based voting rules, are only known to satisfy PR when voters are assumed to sincerely express their preferences.\footnote{For the case of the Single Transferable Vote (STV) rule we refer the reader to~\cite{Dumm84a}, and ~\cite{Wood94a}} However it has long been established that, under very general conditions, every non-dictatorial voting rule is vulnerable to manipulation by strategic, or self-interested, voters~\cite{Gibb73,Satt75}. Furthermore, empirical studies have shown that strategic voting indeed occurs in real-world elections which use PR-based voting rules~\cite{Fred14,GsSt14}. Thus, it is unclear in practice and in the presence of strategic voters whether or not PR-based voting rules satisfy PR.


Designing voting rules which achieve a normative property such as PR in the presence of strategic voters is a mechanism design problem. A typical approach in this literature, is to provide incentives for voters to truthfully reveal their preferences such that an election outcome which satisfies PR can be produced~\cite[Chapter 2]{Park01}. In this case, the PR-based voting rule is said to be \emph{strategyproof} and so, even in the presence of strategic voters, PR is expected to be satisfied. In general the previously mentioned result of Gibbard and Satterthwaite tells us that no such voting rule exists since there is no strategyproof non-dictatorial voting rule, let alone a rule which also satisfies PR. Crucially, this impossibility result of Gibbard and Satterthwaite relies on voters being able to possess arbitrary preferences. In reality however, voter preferences embody greater structure; this has led researchers to study elections and voting rules under assumptions such as correlated preferences~\cite{MaPa16,SBM15}, and with additional restrictions on preferences~\cite{Arro51,Blac48,Sen66}. One particular restriction of preferences is that of \emph{dichotomous} preferences over candidates~\cite{BMS05,BrFi78}, whereby voters view each candidate as either `approved' or `disapproved', but not both.\footnote{Note that a voter with dichotomous preferences over candidates does not, in general, have dichotomous preferences over (multi-winner) election outcomes.} This preference restriction has attracted particular attention since it completely circumvents the impossibility result of Gibbard and Satterthwaite; that is, there exists a non-dictatorial voting rule which is strategyproof~\cite{BrFi78,Webb78}. Unfortunately, even in this restricted setting, there is no strategyproof PR-based voting rule~\cite{Pete17}. This impossibility result of Peters highlights a stark barrier to achieving PR in presence of strategic voters, and suggests the need for an alternate approach which does not require truthful-revelation of preferences by voters.

In this paper we consider an alternate equilibrium approach and show that PR can indeed be achieved in the presence of strategic voters when voters have dichotomous preferences over candidates.\footnote{Justification for this equilibrium approach to achieving PR, over introducing further restrictions on voter preferences, is provided in section 2.1. after the statement of Theorem~\ref{JRnotstrat}.} We design a \emph{strategically robust} PR-based voting rule which is not strategyproof but ensures the existence of at least one (possibly insincere) costly voting equilibrium satisfying PR with respect to voters' private and unrevealed preferences. In contrast, a commonly applied voting rule is shown not to be strategically robust. Furthermore, we prove a limit on `how strategically robust' a PR-based voting rule can be; we show that there is no PR-based voting rule\footnote{We assume the voting rule is deterministic, resolute and anonymous.} which ensures that \emph{every} costly voting equilibrium satisfies PR. Collectively, these results highlight the possibility and limit of achieving PR in the presence of strategic voters and a positive role for mechanisms, such as pre-election polls, which coordinate voter behaviour towards equilibria which satisfy PR.

\paragraph{Contributions}

Our contributions are four-fold; firstly we introduce a new novel approach for the social choice literature to consider (possibly) insincere but desirable election outcomes; secondly, we provide positive results showing the existence of \emph{strategically robust} PR-based voting rules which satisfy PR in the presence of strategic voters; thirdly we provide an impossibility result for the full implementation of proportional representation. This latter result highlights a limit on `how robust' a PR-based voting can be to strategic voters. As a final conceptual contribution, our work sheds light on features of voting rules (such as election thresholds) which lead to strategic robustness and insights into how a given voting rule can be modified to be strategically robust and whilst still maintaining some of the original (sincere voting) properties. 





\paragraph{Outline}

The structure of the paper is as follows; section 2 presents the model and describes our equilibria refinement, and section 3 presents some background for the sincere voting setting. Section 4 includes our key results showing that indeed there exist insincere equilibria which guarantee proportional representation, subsection 4.1. includes our full-implementation impossibility result, the paper then concludes with a brief conclusion. 


\subsection{Related literature}



\noindent \textbf{Proportional representation axioms:} 
 In the context of multi-winner approval elections and voters with dichotomous preferences over candidates, the social choice community has intensely studied proportional representation axioms. In this setting a particularly compelling notion of proportional representation was formulated in~\cite{SFF+17a} called \emph{proportional justified representation} (PJR). This notion of proportional representation is the focus of this paper. Much of the literature has focused on computational concerns about computing election outcomes satisfying these, or similar, axioms when voters are assumed to vote sincerely~\cite{ABC+16a,AEH+18,BFJL16a,SFF+17a}. There also exists a substantial body of research which consider the computational complexity of identifying strategically beneficial actions and equilibria under different voting rules and with different assumptions~\cite{EMOS15a,EMOP16a,EOPR16a,Lee17,MPR08b,MPRZ08a,OMT14a}. However, to the best of our knowledge no attempts have been made to construct or ask whether there exists voting rules which satisfy a desirable notion, such as proportional representation, in the presence of strategic voters and in (possibly) insincere equilibria.\\

\noindent \textbf{Proportional representation in the presence strategic voters:} There are a number of papers which consider the implications of strategic voters in proportional representation voting systems~\cite{AuBa88,BaDi01,CoSh96,DeIa08,DIMP14,SlWh10}, however their research questions and results differ from this paper. In particular, these papers study the observed characteristics and properties of strategic equilibrium under PR systems; whilst, the focus of this paper is whether in equilibrium the unobserved property of PR (with respect to voters' private and unrevealed preferences) can be satisfied; or relatedly, how the observed equilibrium outcomes under PR systems qualitatively differ from the (unobserved) sincere voting outcome.\\

\noindent \textbf{Approval voting and dichotomous preferences:} Approval voting has been subject to intense research in recent years \cite{BrFi07c,LaSk17}. One of the appeals of approval voting is that under dichotomous preferences, a simple approval-based voting rule is strategyproof~\cite{BrFi78,Webb78}, hence escaping the impossibility results of~\cite{Gibb73,Satt75}. However, as shown by~\cite{Fish78,Vors07} the set of approval-based voting rules which are strategyproof is limited -- in particular, under mild conditions the approval voting rule can be completely and uniquely characterised by whether or not it is strategyproof. More recently, and of greater relevance to this paper,~\cite{Pete17} showed that a very weak form of proportional representation is incompatible with strategyproofness when considering deterministic voting rules. That is, there is no deterministic voting rule which is proportionally representative and strategyproof. Our research direction complements and contrasts this impossibility result by showing that, in fact, strategyproofness is not necessarily required to ensure proportional representation in equilibrium.  \\

\noindent \textbf{Implementation theory:} There is an extensive literature on the implementation of social choice functions/correspondences via equilibrium concepts such as Nash equilibria. The canonical \emph{full implementation} question is whether a mechanism can be designed such that every equilibrium produces an outcome which is considered `desirable' with respect to a given social choice function or correspondence -- however this literature does not provide answers to the questions posed and solved in this paper. In particular, the structure of mechanisms typically used in the implementation theory literature are not suited to standard voting environments. This is sometimes due to the intricacies of the mechanism or the informational requirements. For example, the classical literature pioneered by Maskin~(see \cite{Mask77, Mask85}) assumes that every player knows every other players' ``private" information, i.e., the entire profile of voter preferences.\footnote{In this classical setting, the players' preferences are common knowledge among other players but is ``private" and unknown from the perspective of the mechanism designer.} In the context of voting this is indeed an extreme assumption which we do not apply in our model. Such  informational requirements have been relaxed by a number of scholars~\cite{AbMa90,Jack92}, but still the mechanisms are not appropriate for the environment and focus of this paper. Despite this, and in line with the canonical full implementation question, we provide an impossibility of full implementation result under assumptions appropriate for standard voting environments. In this environment, we show that it is impossible to fully implement the social choice correspondence of proportional representation.

\section{The model, equilibria and costly voting}



\subsection{Model}
%
%
%

We introduce the following social choice model, let $N$ be a set of voters, $C$ be a set of candidates.\footnote{This model can be easily reformulated as a party-seat model by partitioning the candidate set into parties and deriving voter preferences over candidates from their preferences over parties.} We assume that voters have dichotomous preference over candidates; for a given voter $i\in N$, let $A_i\subseteq C$ denote voter $i$'s true preference (i.e., most preferred equivalence class of candidates), and hence $C\backslash A_i$ is the set, or equivalence class, of voter $i$'s unapproved, or disapproved, candidates. Let $\hat{A}_i$ be the reported (not necessarily truthful) approval ballot. If $\hat{A}_i=A_i$ then we say the reported ballot is \emph{truthful}, or \emph{sincere}, otherwise, $\hat{A}_i\neq A_i$ and we say that the reported ballot $\hat{A}_i$ is \emph{insincere}. We denote the profile of preferences and reported ballots by $A$ and $\hat{A}$; that is, $A=(A_1, \ldots, A_n)$ and $\hat{A}=(\hat{A}_1, \ldots, \hat{A}_n)$, respectively.

The goal of the election is to select a winning committee $W\subseteq C$ of some predetermined size $k<|C|$ (where $k$ is a positive integer). The value $k$ is common knowledge. We assume voter $i$'s (ordinal) preferences over election outcomes are as follows: given two committees $W$ and $W'$, 
\begin{align}\label{equation: preferences}
W\pref_i W' \iff |W\cap A_i|\ge |W'\cap A_i|.
\end{align}
The strict preference $\spref_i$ can be derived from the strict inequality. We emphasise that despite voters having dichotomous preferences over candidates, given~(\ref{equation: preferences}), voter preferences over election outcomes need not be dichotomous.

If voters face election outcomes which are random, the voter must compare probabilistic outcomes, or lotteries, of election outcomes. We assume voter preferences over lotteries are given by the upward-lexicographic extension ($ul$-extension). The $ul$-extension relates to the lexicographic preferences introduced by~\cite{Haus54}. Given the preferences of a voter over deterministic election outcomes, the $ul$-extension derives preferences over lotteries by comparing relative probabilities of the worst case outcomes in a lexicographic manner; that is, the voter prefers lotteries which minimise the probability of `bad' outcomes (we delay the formal definition for now). 

The $ul$-extension is a strong assumption to apply however when considering extensions of ordinal preferences to lotteries there are limited choices with the appropriate properties for equilibrium analysis. For example, the commonly applied first-order stochastic dominance extension ($sd$-extension) is not \emph{complete}; that is, two lotteries may not be comparable. Thus if equilibrium analysis were considered with the $sd$-extension, voters would only consider the restricted set of strategies which lead to comparable lotteries and hence many equilibria would arise simply because of the incompleteness of the preference extension. Whilst the $ul$-extension is complete, and hence voters can always compare the lotteries induced by any strategy deviation. Furthermore we note that any lottery preferred to another lottery under the $sd$-extension, is also preferred under the $ul$-extension. Thus, all of our equilibria results which assume the $ul$-extension also hold under the $sd$-extension. We refer the reader to~\cite{Cho12} for a detailed discussion of the $ul$-extension, other related extensions, and their properties.
%
%
%
%
%
%


We now formalise the $ul$-extension with respect to the preference $\pref_i$ over (deterministic) election outcomes introduced earlier. A lottery is a probability distribution, $L$, over the set of possible election outcomes $W$. Let $\pi_L^{(i)}(j)$ denote the probability that the lottery $L$ produces an election outcome $W$ such that $|W\cap A_i|=j$. Note that $\sum_{j=0}^{\min\{k, |A_i|\}} \pi_L^{(i)}(j)=1$  for all $i\in N$.
 
A voter weakly prefers a lottery $L$ over another lottery $L'$ if and only if;
\begin{enumerate}
\item there exists an integer $ \ell\in \{0, 1, \ldots, k\}$ such that $\pi_L^{(i)}(j)=\pi_{L'}^{(i)}(j)$ for all $j< \ell$ and $\pi_L^{(i)}(\ell)<\pi_{L'}^{(i)}(\ell)$; or
\item $\pi_L^{(i)}(j)=\pi_{L'}^{(i)}(j)$ for all $j$.
\end{enumerate}
Abusing notation slightly we denote the preference relation over lotteries by $\pref_i$. A voter strictly prefers $L$ over $L'$, i.e. $L\spref_i L'$, if $L\pref_i L'$ and there exists some integer $j\in \{0, \ldots, k\}$ such that $\pi_L^{(i)}(j)\neq \pi_{L'}^{(i)}(j)$.


\begin{remark}
The $ul$-extension of preferences conveys a high level of risk-aversion. When considering a more risk-averse preference extension, such as worst-case preferences\footnote{i.e., a voter values a lottery based on the least preferred outcome.}~\cite{Alon15}, then all of our results still hold. In fact, in some cases stronger results can be attained. This suggest that a higher degree of risk aversion is conducive to the goals of this paper; in particular, the construction of strategically robust proportional representation voting rules. The opposite can be said for risk-loving preference extensions.
\end{remark}

We now introduce the proportional representation axiom of proportional justified representation (PJR), this axiom was considered in~\cite{AEH+18,BFJL16a,SFF+17a}. PJR captures the idea that cohesive groups of voters, who unanimously agree on some set of candidates as being `appproved', should receive representation in proportion to their size. 

\begin{definition}\emph{[Proportional Justified Representation (PJR)]}\label{definition: PJR} Given the (truthful) ballot profile $A=(A_1, \ldots, A_n)$ over a candidate set $C$ and a target committee of size $k$, we say that a set of candidates $W$ of size $k$ satisfies proportional justified representation (PJR) for $(A,k)$ if for every positive integer $\ell \le k$, $\forall X\subseteq N$:
$$|X|\ge \ell n/k\quad \text{and}\quad |\cap_{i\in X} A_i|\ge \ell \implies |W\cap(\cup_{i\in X} A_i)|\ge \ell.$$
In the case of a (random) lottery $L$ over election outcomes, we say that PJR is satisfied by $L$ if every non-zero probability election outcome satisfies PJR. That is, PJR is satisfied with probability one.
\end{definition}


%

We shall assume throughout this paper that the number of candidates to be elected $k$ divides the number voters $n$. This assumption avoids the complications surrounding non-integer values of $n/k$ in the proportional representation definition (Definition~\ref{definition: PJR}). For general values of $k$ and $n$, and under sincere voting, PJR is considerably more difficult to satisfy -- in fact, it was posed as an open problem in~\cite{SFF+16} and eventually resolved independently in~\cite{BFJL16a} and~\cite{SFF16a}. We formalise the assumption below, but will always state our result explicitly with the condition that `$k$ divides $n$' when required.


\begin{assumption}\emph{[`$k$ divides $n$']}\label{assumption: k divides n}
Assume that the number of candidates to be elected, $k$, divides the number of voters, $n$. 
\end{assumption}

\begin{remark}
For general values of $k$ and $n$, all of our results follow immediately for a slightly weaker version of the PJR axiom. The weaker version specifies that cohesive groups of voters $X$ of size $\ell \ceil{n/k}$ require $\ell$ representatives, i.e., replace $|X|\ge \ell n/k$ with $|X|\ge \ell \ceil{n/k}$ in the (Definition~\ref{definition: PJR}).
\end{remark}

The following result states that a PJR outcome exists for all preference profiles, we refer the reader to~\cite{SFF+17a} for a proof.

\begin{theorem}\emph{\cite{SFF+17a}}
A PJR outcome always exists.
\end{theorem}

Given a voting rule which takes $(\hat{A},k)$ as an input we denote the (possibly degenerate) lottery of election outcomes as $L_{\hat{A}}$. Note that $\hat{A}$ is the reported ballot profile and need not be truthful, i.e. it may be the case that $\hat{A}\neq A$.

\begin{definition}\emph{[strategyproof]}
A voting rule is said to be strategyproof if for all $i\in N$ 
$$L_{(A_i, \hat{A}_{-i})}\pref_i L_{(\hat{A}_i, \hat{A}_{-i})}\qquad \text{for all\, } \hat{A}_i \text{ and for all }\hat{A}_{-i}.$$
\end{definition}

As mentioned in the introduction,~\cite{Pete17} showed that there is in fact no (deterministic and resolute\footnote{A resolute voting rule produces a single election outcome, rather than a set of outcomes.}) voting rule which is strategyproof and satisfies a very weak notion of proportional representation. An immediate implication of the results is that PJR is incompatible with strategyproofness.


\begin{theorem}[\cite{Pete17}]\label{JRnotstrat} 
If $|C|, |N|\ge 4$ and $k\ge 2$ there exists no deterministic and resolute voting rule that is strategyproof and satisfies PJR. 
\end{theorem}


The above result highlights the necessity for an alternate approach to achieving proportional representation, or the introduction of additional assumptions to circumvent this impossibility result. In the social choice literature for similar, but distinct, impossibility results there is a large body of work which consider restrictions of voter preferences to achieve possibility results. Some preference restrictions include single-peakedness of preferences~\cite{Arro51,Blac48}, or correlations of voter preferences~\cite{SBM15,MaPa16}. In this paper, we consider an alternate approach by moving to equilibrium analysis instead. 

This equilibrium approach is justified on the basis that; (1) the axiom of proportional representation is focused on the idea that voters have diverse preferences, and the difficulty of satisfying this axiom is due to precisely this diversity. Thus, any restriction to the diversity of voter preferences is weakening both the constraint imposed by the PR axiom and also the conflict between voter strategies and preferences -- in a sense this approach is `giving up too much'. (2) the impossibility result of~\cite{Pete17} holds for a setting where voter preferences are already assumed to be restricted to dichotomous preferences; thus, it is hard to justify even further preference restrictions independent of the PR axiom. (3) The restriction of preferences to single-peakedness~\cite{Arro51,Blac48}, or value-restricted preferences~\cite{Sen66}, though at times seemingly innocuous, has little empirical support~\cite{Matt11a}.


Before concluding this subsection, we note that impossibility result of~\cite{Pete17} has only been proven for deterministic voting rules. Indeed in various social choice settings it has been shown that randomised voting rules can sometimes be both necessary and sufficient to guarantee strategyproofness without resorting to a dictatorship; for example see~\cite{ABBM14,AzYe14,CLPP13,Proc10}. Thus, when considering randomised voting rules is it unclear whether there is an incompatibility between strategyproof and PJR.



\subsection{Equilibria and costly voting}


This subsection formalises what is meant by a (pure) Nash equilibrium and describes the strategic, or behavioural, effect of a costly voting environment. The costly voting environment leads to 
additional structure on voter preferences which are similar to those attained by the models considered in~\cite{DeIa05,PaRo83,EMOS15a}.\footnote{In section 4 we define a costly voting refinement of the strong Nash equilibrium concept.}


\begin{definition}\emph{[Best response]}
If a voter $i\in N$ faces a reported ballot profile $\hat{A}_{-i}$ and 
$$L_{(\hat{A}_i, \hat{A}_{-i})}\pref_i L_{(\hat{A}_i', \hat{A}_{-i})}\quad \text{ for all $\hat{A}_i'$,}$$ then $\hat{A}_i$ is a best response for voter $i$.
\end{definition}

\begin{definition}\emph{[Pure-Nash Equilibria (PNE)]}
A reported ballot profile $\hat{A}$ is a pure-Nash equilibrium (PNE) if for every voter $i\in N$ the ballot $\hat{A}_i$ is a best response to $\hat{A}_{-i}$. If the lottery $L$ is the output of a PNE reported profile $\hat{A}$, then we say that $L$ is supported by a PNE. 
\end{definition}


We also define what is meant by a pivotable voter. Given a voting rule and a voting instance; informally, a voter is said to be pivotable if they can change the election outcome by changing their reported ballot. Thus, when checking whether a profile of reported ballots is a PNE it suffices to consider only pivotable voters.

\begin{definition}\emph{[Pivotal voter]}
Given a voting rule and a reported ballot profile $\hat{A}$, a voter $i\in N$ is said to be pivotable if there exists a report $\hat{A}_i'\neq \hat{A}_i$ such that 
$$L_{\hat{A}}\neq L_{(\hat{A}_i', \hat{A}_{-i})}.$$
Otherwise, the voter is said to be non-pivotable. 
\end{definition}

We now define how voters behave in a costly voting environment and define the corresponding equilibrium concept.\footnote{ In the literature this behavioural tendency has been referred to a lazy voting~\cite{EMOS15a}, voting with abstention and also costly voting.} Consideration of a costly voting environment is motivated by the fact that voting is in fact a costly activity, and applying game-theoretic analysis to voting games with costless voting leads to the existence of equilibria which are not observed in practice. The costly voting environment can be modelled by the introduction of a small cost experienced by voters who participate in the election. It is important that the cost is sufficiently small to ensure that only voters who are indifferent between the election outcome when they participate and when they abstain are incentivised to abstain.

\begin{definition}\emph{[Costly voting]}
Under costly voting, if a voter $i\in N$ faces a reported ballot profile $\hat{A}_{-i}$ and $L_{(\emptyset, \hat{A}_{-i})}\pref_i L_{(\hat{A}_i', \hat{A}_{-i})}$ for all $\hat{A}_i'$,
then voter $i$ will optimally report $\hat{A}_i=\emptyset$. 
\end{definition}

\begin{definition}\emph{[Costly voting equilibrium]}
A reported ballot profile $\hat{A}$ is a costly voting equilibrium if $\hat{A}$ is a pure-Nash equilibrium and voters experience costly voting.
\end{definition}


\section{Achieving proportional representation under sincere voting}


This section introduces a strategically robust proportional representation voting rule, called GreedyMonroe, and states properties of the voting rule which hold under sincere voting. We show that under sincere voting the rule satisfies proportional justified representation (PJR), and is not strategyproof. 




The original GreedyMonroe voting rule was introduced by~\cite{SFS15a}, and further studied by~\cite{EFSS17a,SFF+17a}. We introduce a variant of the GreedyMonroe voting rule which maintains the key representation properties of the original rule but ensures strategic robustness, for simplicity we still refer to the voting rule as GreedyMonroe.

 
Let $k$ denote the number of candidates to be elected, and $n$ the number of voters. Our version of the GreedyMonroe voting rule asks each voter to report their dichotomous preferences via an approval ballot $\hat{A}_i$, then proceeds as follows: In the first-stage, the rule considers the set of candidates $C^*$ who have approval support $\ge n/k$ and then elects a candidate uniform at random from $C^*$ and removes $\ceil{n/k}$ approval ballots of voters who supported this now-elected candidate. With the updated approval ballot profile the rule iteratively repeats the previous step until either a committee of size $k$ is elected, or $<k$ candidates have been elected and no candidate receives approval $\ge n/k$. In the latter case, a second-stage is used to uniformly at random select the remaining candidates from the set of unelected candidates, i.e., regardless of the relative approval scores.

Note that the voting rule ensures that every elected candidate $c$ in the first-stage represents at least $n/k$ distinct voters, and any voters who gave support for candidate $c$ beyond $\ceil{n/k}$ is able to contribute their vote to another so-far-unelected candidate on their ballot. This process bears some resemblance to a discrete re-weighting version of the well-known Single Transferable Vote (STV), which is defined for strict-order preferences and used world-wide for achieving proportional representation in national elections.\footnote{STV is used in countries such as Australia, India, Ireland and Pakistan. See within~\cite{AzLe17b} for a general formulation of the STV voting rule and brief overview of its history.}

The voting rule breaks ties using a uniform random tie-breaking rule. That is, given a tie between a candidate set $C^*\subseteq C$, each candidate is selected with probability $\frac{1}{|C^*|}$ -- the selected candidate is the `winner' of the tie-break. We provide a formal description of the voting rule in Algorithm~\ref{algo:Greedy Monroe}.


		\begin{algorithm}[]
								  \caption{(modified) GreedyMonroe algorithm}
								  \label{algo:Greedy Monroe}

								\begin{algorithmic}
									\REQUIRE  Election size $k$, candidate set $C$, reported ballot profile $\hat{A}$.
																		\ENSURE A winning committee $W\subseteq C$ of size $k$.
								\end{algorithmic}
								\begin{algorithmic}[1]
									\STATE Initialise $\tilde{W}=\emptyset, C'=C, \hat{A}'=\hat{A}$.								
									\STATE  Let $C^*$ be the set of candidates $c'\in C'$ such that $s(c', \hat{A}'):=|\{i\in N\, :\, c'\in \hat{A}_i'\}|\ge n/k.$
									\STATE If $|\tilde{W}|<k$ and $C^*\neq \emptyset$, select some candidate $c^*\in C^*$ by a uniform random tie-breaking rule. Set $\tilde{W}\mapsto \tilde{W}\cup c^*$, then remove precisely $\ceil{n/k}$ ballot profiles from $\hat{A}'$ with $c^*\in\hat{A}_i'$ (choosing voters uniformly at random) and redefine $C'\mapsto C'\backslash \{c^*\}$, and then repeat step 2.
									\STATE		If $C^*=\emptyset$ or $|\tilde{W}|<k$, then select the remaining $(k-|\tilde{W}|)$ candidates from $C\backslash \tilde{W}$ via a uniform random tie-breaking rule, to form a committee $W$ of size $k$ such that $\tilde{W}\subseteq W$.				
%
%
%
%
								\end{algorithmic}
							\end{algorithm}
%
%
%

We highlight three strategically important features of the GreedyMonroe voting rule (Algorithm~\ref{algo:Greedy Monroe}): (1) the removal of ballots after the election of candidates means that each voter's ballot is only considered once, and hence approving of more than one candidate is never strictly optimal over approving a single candidate; (2) the sharp threshold of $n/k$ creates an incentive for voters to not provide additional support to candidates who already have $\ge n/k$ votes, (3) the quota-filling uniform random selection of candidates creates a lottery over election outcomes. This lottery acts as a `punishment' which creates an incentive for voters to coordinate themselves into group of $\ceil{n/k}$ as to avoid (or minimise) the risk of the partially random election outcome. 


In the special case that `$k$ divides $n$', i.e., the number of candidates to be elected $k$ divides the number of voters $n$, then the GreedyMonroe rule satisfies desirable proportional representation properties. It was shown in~\cite{SFF+17a}, that when `$k$ divides $n$'  voting rules such as GreedyMonroe satisfies PJR when voters report sincerely. We state the result below and refer the reader to~\cite{SFF+17a} for details of the proof.




\begin{proposition}\emph{\cite{SFF+17a}}\label{proposition: GreedyJR JR and not PJR}
If `$k$ divides $n$', GreedyMonroe satisfies PJR when voters are sincere. 
\end{proposition}

As stated in the previous section, a deterministic voting rule which satisfies PJR and is strategyproof does not exist~\cite{Pete17}. This result equally applies to the case where `$k$ divides $n$'.  The GreedyMonroe voting rule satisfies PJR (under sincere voting), however, it is not deterministic. Nonetheless, we prove in the following theorem (via an example) that it is still not strategyproof -- the proof can be found in the appendix.

%

\begin{theorem}\label{theorem: GreedyJR not strategyproof}
The GreedyMonroe is not strategyproof. 
\end{theorem}
%
%
%


\subsection{Does the voting rule matter?}


It is well-known that, in elections with more than two voters and a standard voting rule, any election outcome can be supported by a PNE~\cite[Footnote 1]{DeIa05}. Thus a natural questions arises ``Does the voting rule matter?" since we are interested in constructing \emph{strategically robust} voting rules which attain PR in equilibrium. We argue that indeed the voting rule does matter. It has been noted by a number of scholars that considering elections within a costly voting environment is essential to attain predictive power from equilibria analysis. Thus our focus is on costly voting equilibria, and within this environment it is no longer true that, under standard voting rules, any election outcome can be supported by a PNE. Furthermore and as a point of illustration the most commonly applied voting rule for dichotomous preferences, referred to as the AV-rule or simply approval voting,\footnote{see~\cite{BrFi78,Webb78}.} can be shown to fail our definition of \emph{strategic robustness}; that is, there need not be any costly voting equilibria which satisfy PJR. In fact the AV-rule need not have any costly voting equilibrium -- we provide a formal statement and proof of this in the appendix.

\section{Insincere but desirable equilibria: Existence}








We now prove the existence of a PJR outcome supported by a PNE, under the GreedyMonroe voting rule. This PNE, and details within the proof, will be integral to extending the result to the more appropriate costly voting equilibrium concept. It is important to note that the PNE constructed in this theorem have additional desirable features; such as the fact that a voter is never required to report to approve of a candidate which they do not truthfully approve of. 

The intuition of the theorem can be understood as follows. Consider a two-stage process; voters first observe a perfectly accurate pre-election poll which announces the guaranteed winners from a sincere election outcome based on the poll, then voters sequentially remove their votes for any unelected/losing candidate(s). In addition, since the GreedyMonroe voting rule considers each ballot at most once, every voters removes any `excess' votes for a winning candidate\footnote{This type of behaviour whereby voters remove their vote from candidates with excess support to other candidates (or decide to abstain) is analogous to the ``Chicken" effect referred to by~\cite{CoSh96}.} -- thus, restricting their vote to at most one (winning) candidate in election. Lastly, there exists at least one such partition which ensures the actual outcome still coincides with the outcome as per the poll announcement.

\begin{theorem}\label{theorem: k divides n greedy leads to PJR}
Let the voting rule be GreedyMonroe, and assume $k$ divides $n$.  There exists a PNE $\hat{A}$ supporting a PJR outcome with respect to $A$, such that $\hat{A}_i \subseteq A_i$ for all $i\in N$. Furthermore, this PNE coincides with the sincere voting outcome. 
\end{theorem}

\begin{proof}
Let $\tilde{W}=\{c_1, \ldots, c_j\}$ be the set of $j\le k$ candidates elected in the first-stage of the GreedyMonroe rule when voters are sincere. For each candidate $c_t\in \tilde{W}$ let $N_t$ denote the set of $n/k$ voters who unanimously and truthfully supported candidate $c_t$, and were then selected to have their ballots removed (line 3 of the voting rule). 

The reported ballot profile $\hat{A}$ which will be shown to be a PNE is constructed as follows: For each $t\in \{1, \ldots, j\}$, set $\hat{A}_i=\{c_t\}$ for all $i\in N_t$ (note that there are precisely $n/k$ such voters in each $N_t$). For all other voters set $\hat{A}_i=\emptyset$.

At the termination of this procedure we have a reported ballot profile $\hat{A}$ such that all candidates in $\tilde{W}$ receive precisely $n/k$ reported approvals and all other candidates receive zero. Thus, the only pivotable voters who can change the election outcome are those such that $\hat{A}_i=\{c\}$ for some candidate $c$ (since $n/k>1$). Furthermore, for all voters $i\in N$ we have $\hat{A}_i\subseteq A_i$ -- hence voters approve of their reported approval (w.r.t. true preferences $A_i$). 

Also note that the outcome of $\hat{A}$ under GreedyMonroe is a set $W'$ which contains $\tilde{W}$ with probability one (i.e. all candidates in $\tilde{W}$ are necessarily selected). Thus with probability one PJR is satisfied w.r.t. the true preferences $A$. 

It remains to show that this is a PNE. Suppose that a voter $i\in N$ has a (strictly) profitable deviation $\hat{A}_i'\neq \hat{A}_i$, then it must be the case that this voter is pivotable and so $\hat{A}_i=\{c\}$. There are two cases to consider which can arise when $\hat{A}_i'$ is submitted.

Case 1: Nothing happens with probability one. This occurs if $c\in \hat{A}_i'$ and there is no other candidate $c'\in \hat{A}_i'$ and also $c'\in \tilde{W}$ -- thus, $\tilde{W}$ is still elected with probability one, and the distribution/lottery of election outcomes is unchanged. This contradicts the assumption that $\hat{A}'$ is (strictly) profitable, since the voter is indifferent between both ballots $\hat{A}_i$ and $\hat{A}_i'$.

Case 2: Candidate $c$ is removed from the election outcome with non-zero probability in line 3 of the voting rule. This can occur if $c\notin \hat{A}_i'$ and/or there is another candidate $c'\in \hat{A}_i'$ and also $c'\in \tilde{W}$. We now show that this can never be strictly profitable. 

Within Case 2, in the non-zero probability event that candidate $c$ is not elected in line 3 of the voting rule, the rule selects $k-|\tilde{W}|+1$ candidates uniformly at random from the set of unelected candidates $(C\backslash\tilde{W})\cup\{c\}$. Note that alternatively when $c$ is elected in line 3 (or under the equilibrium profile $\hat{A}$), the rule selects just $k-|\tilde{W}|$ candidates uniformly at random from the set of unelected candidates $C\backslash\tilde{W}$. 

To analyse this case we must consider the number of candidates which voter $i$ disapproves of in the set $(C\backslash \tilde{W})\cup\{c\}$, we denote this by
$$\ell=|\big((C\backslash \tilde{W})\cup\{c\}\big) \backslash A_i|.$$
We assume that $\ell>0$, otherwise voter $i$ is indifferent between the lotteries induced by $\hat{A}_i$ and $\hat{A}_i'$ which contradicts the assumption that $\hat{A}_i'$ is strictly profitable.

Now if $\ell\ge k-|\tilde{W}|+1$ then the worst-case (non-zero probability) outcome produced from $\hat{A}_i'$ includes precisely one less (truthfully) approved candidate in the winning committee. This follows since when voter $i$ submits $\hat{A}_i'$ with non-zero probability the candidates in $\tilde{W}\backslash \{c\}$ are elected along with $k-|\tilde{W}|+1$ disapproved candidates -- whilst, under $\hat{A}_i$ the worst-case outcome includes $\tilde{W}$ and just $k-|\tilde{W}|$ disapproved candidates.\footnote{Note that candidate $c\in \tilde{W}$ is truthfully approved by voter $i$, i.e., $c\in A_i$.} Thus, we conclude that voter $i$ strictly prefers $\hat{A}_i$ over $\hat{A}_i'$ (under the $ul$-extension of preferences to lotteries). 

Now suppose that $0<\ell< k-|\tilde{W}|+1$, and again denote the lotteries under $\hat{A}_i$ and $\hat{A}_i'$ by $L$ and $L'$. The worst-case outcomes $W$ and $W'$ induced by $L$ and $L'$, respectively, are equally preferred and are attained under either lottery with non-zero probability when all $\ell$ disapproved candidates from $C\backslash \tilde{W}$ are elected. Denote the number of candidates involved in the random lottery $L$  by $m=|C\backslash \tilde{W}|$, and denote the number of candidates to be elected in this lottery by $j=k-|\tilde{W}|$. The probabilities of these worst-case outcomes being achieved (conditional on candidate $c$ being unelected in lines 3 of the voting rule) is given by the hypergeometric distribution. 

To see this consider the following reformulation; under lottery $L$, the worst-case outcome is when, with the $j$ selection of candidates among $m$ choices, all $\ell$ disapproved candidates are elected. This is analogous to the `urn problem' of calculating the probability that from an urn, with $\ell$ blue balls and $m-\ell$ white balls, $j$ draws (without replacement) leads to all $\ell$ blue balls being drawn. In particular, this probability is given by the hypergeometric distribution (see within~\cite{JoKo77} for details):
\begin{align*}
\mathbb{P}_L^{(i)}(\ell)=\frac{{\ell \choose \ell}{m-\ell \choose j-\ell }}{{m\choose j}}=\frac{{m-\ell \choose j-\ell }}{{m\choose j}}
\end{align*}
whilst, under $L'$ we have $\mathbb{P}_{L'}^{(i)}(\ell)={{m+1-\ell \choose j+1-\ell }}/{{m+1\choose j+1}}$. This follows since the lottery under $L'$, $(j+1)$ draws are made from the set of $(m+1)$ candidates $(C\backslash \tilde{W})\cup\{c\}$ which still only includes $\ell$ disapproved candidates (recall that $c\in A_i$). However, notice that 
\begin{align*}
\mathbb{P}_{L'}^{(i)}(\ell)&=\frac{{m+1-\ell \choose j+1-\ell }}{{m+1\choose j+1}}=\frac{\frac{m+1-\ell}{j+1-\ell}{m-\ell \choose j-\ell }}{\frac{m+1}{j+1}{m\choose j}}=\frac{\frac{m+1-\ell}{j+1-\ell}}{\frac{m+1}{j+1}}\mathbb{P}_L^{(i)}(\ell).
\end{align*}
and for positive $m,j,\ell$ such that $j+1>\ell$ we have
\begin{align*}
\frac{m+1-\ell}{j+1-\ell}\frac{j+1}{m+1}>1\iff (m+1-\ell)(j+1)>(j+1-\ell)(m+1)\iff \ell(m-j)>0,
\end{align*}
which clearly holds since $0<\ell<j+1\le m$. Thus, we have $\mathbb{P}_{L'}^{(i)}(\ell)>\mathbb{P}_L^{(i)}(\ell)$.

Thus, the probability (conditional that $c$ was not elected in line 3 of the voting rule) that the $\ell$ disapproved candidates are elected from the (random quota-filling) draw induced by $L$ is strictly less than that under $L'$.

Now denote the number of approved candidates in the worst-case outcomes $W$ and $W'$ induced by $L$ and $L'$, respectively, by $\alpha$ (note that both $W$ and $W'$ are equally preferred and so have the same number of approved candidates). Under the $ul$-extension, if the worst-case outcomes are equally preferred under both lotteries we compare the (unconditional) probabilities that this worst-case outcome is attained, i.e., $\pi_L^{(i)}(\alpha)$ and $\pi_{L'}^{(i)}(\alpha)$. Under the lottery $L'$, there are two non-zero probability events; (1) candidate $c$ is elected in lines 3 of voting rule and then $k-|\tilde{W}|$ candidates are elected from $C\backslash \tilde{W}$, or (2)  candidate $c$ is not elected in lines 3 of the voting rule and then $k-|\tilde{W}|+1$ candidates are elected from $(C\backslash \tilde{W})\cup\{c\}$. Denote the probabilities of these events by $p_1$ and $p_2$, respectively, note that $0<p_1<1$ and $p_2=1-p_1$. 

For lottery $L$', conditional of event (1), the probability that the worst-case outcome is achieved is given by $\mathbb{P}_L^{(i)}(\ell)$ (i.e. the same as under lottery $L$) and conditional of event (2), the probability that the worst-case outcome is achieved is given by $\mathbb{P}_{L'}^{(i)}(\ell)$. Note that the probability that the worst-case outcome is attained under $L$ is simply $\mathbb{P}_L^{(i)}(\ell)$, since candidate $c$ is elected at line 3 of the voting rule with probability one. Thus, we have
\begin{align*}
\pi_{L'}^{(i)}(\alpha)&=p_1 \mathbb{P}_L^{(i)}(\ell)+p_2 \mathbb{P}_{L'}^{(i)}(\ell)\\
&>p_1 \mathbb{P}_L^{(i)}(\ell)+p_2 \mathbb{P}_{L}^{(i)}(\ell) &&\text{since } \mathbb{P}_{L'}^{(i)}(\ell)>\mathbb{P}_{L}^{(i)}(\ell)\\
\implies  \pi_{L'}^{(i)}(\alpha)&>\pi_{L}^{(i)}(\alpha).
\end{align*}
That is, the worst-case outcome (which is equivalent under $L$ and $L'$) is strictly more likely under $L'$. We conclude that voter $i$ strictly prefers the lottery $L$ over $L'$, and hence we have a contradiction since $\hat{A}_i'$ can not be a strictly profitable deviation.
Thus, the reported ballot profile $\hat{A}$ is a PNE and the election outcome $W$ satisfies PJR with respect to the true preferences.
\end{proof}

%


Turning our focus to the costly voting equilibrium concept, we identify a sufficient condition to ensure that the PJR equilibrium outcome from the previous theorem which coincides with the sincere outcome is also a costly voting equilibrium (Theorem~\ref{theorem: costly voting PJR under greedyMonroe}). The sufficiency condition (Assumption~\ref{assumption: sufficiency}) is then shown to be `tight' in the sense that in general without the sufficiency condition the sincere voting outcome need not be possible to sustain in a costly voting environment (Theorem~\ref{theorem: tightness of sufficiency condition}).


For future reference, we formally define this `sufficiency condition' as an assumption but it will be made explicit in the statements of results whether this condition is indeed required. The condition requires that every candidate (truthfully) disapproves of at least $k$ candidates. Informally, this condition ensures that there are enough `disapproved' candidates to incentivise voters to vote according to a costly voting equilibrium, and avoid the random quota-filling feature of the voting rule. 

\begin{assumption}\emph{[Sufficiency condition]}\label{assumption: sufficiency}
We say that the `sufficiency condition' holds if every voter disapproves of at least $k$ candidates in $C$ with respect to their true preferences; that is, 
$$|C\backslash A_i|\ge k \qquad \text{for all voters $i\in N$.}$$
\end{assumption}

\begin{remark}\label{remark: suff}
Under a party-seat reformulation of the election model whereby $p$ political parties (each with at least $k$ candidates) contest for seats and voters cast votes based on a candidate's party, the sufficiency condition simply requires that each voter disapproves of at least one party.
\end{remark}




With this sufficiency condition (Assumption~\ref{assumption: sufficiency}), the result of Theorem~\ref{theorem: k divides n greedy leads to PJR} can be maintained in a costly voting environment. That is, a PJR outcome can be supported by a costly voting equilibrium. This is achieved by ensuring that every non-abstaining voter is pivotable and strictly prefers to participate than to abstain.

\begin{theorem}\label{theorem: costly voting PJR under greedyMonroe}
Let the voting rule be GreedyMonroe, and assume $k$ divides $n$. If every voter $i\in N$ disapproves of at least $k$ candidates in $C$ (i.e. Assumption~\ref{assumption: sufficiency} holds) then there exists a costly voting equilibrium $\hat{A}$ supporting a PJR outcome such that $A_i\subseteq \hat{A}_i$ for all $i\in N$. Furthermore, this costly voting equilibrium coincides with a sincere voting outcome under the GreedyMonroe rule. 
\end{theorem}


\begin{proof}
Let $\hat{A}$ be the PNE equilibrium constructed in Theorem~\ref{theorem: k divides n greedy leads to PJR}. Let $\tilde{W}$ be the associated subset of candidates elected in the first-stage of GreedyMonroe under $\hat{A}$. 

Given that $\hat{A}$ is a PNE, it only remains to prove that this is indeed a costly voting equilibrium. That is, it suffices to show that abstaining is never weakly preferred by a non-abstaining voter. Let voter $i\in N$ be such that $\hat{A}_i=\{c\}$, if voter $i$ were to abstain from voting then the worst-case (non-zero probability) election outcome would exclude the (truthfully) approved candidate $c$ from the election outcome and include $\tilde{W}\backslash \{c\}$ and $k+1-|\tilde{W}|$ disapproved candidates from $C\backslash \tilde{W}$ (such a number of disapproved candidates in $C\backslash \tilde{W}$ is guaranteed to exist by Assumption~\ref{assumption: sufficiency}. Whilst when voter $i$ does not abstain and reports $\hat{A}_i=\{c\}$ the worst-case election outcome includes $\tilde{W}$ and just $k-|\tilde{W}|$ disapproved candidates -- this means that precisely one additional approved candidate is included in the worst-case outcome compared to the worst-case outcome when she abstains.\footnote{Note that this corresponds precisely to the case where $\ell\ge k-|\tilde{W}|+1$ within the proof of Theorem~\ref{theorem: k divides n greedy leads to PJR}.} Thus, the voter strictly prefers to vote than to abstain and the ballot profile is a costly voting equilibrium.

The final statement follows immediately since the costly voting equilibrium $\hat{A}$ is precisely the same as the PNE in Theorem~\ref{theorem: k divides n greedy leads to PJR} which supports a PJR outcome which coincides with the sincere voting outcome.
\end{proof}

In general however, the above theorem does not hold. That is, the sufficient condition is `tight'.

\begin{theorem}\label{theorem: tightness of sufficiency condition}
Consider the GreedyMonroe voting rule and suppose `$k$ divides $n$'. If every voter $i\in N$ does not disapprove of at least $k$ candidates in $C$ (i.e., Assumption~\ref{assumption: sufficiency} does not hold) then, in general, a PJR outcome which coincides with the sincere outcome can not be supported as a costly voting equilibrium. That is, the sufficiency condition in Theorem~\ref{theorem: costly voting PJR under greedyMonroe} is `tight'.
\end{theorem}

\begin{proof}
Consider the following example; let $|N|=4, k=2, C=\{a,b,c\}$ and the preferences of voters be as follows:
\begin{align*}
A_1=\{a,b,c\}, A_2=\{a\},A_3= A_4=\emptyset.
\end{align*}
Under GreedyMonroe with sincere voting the election outcome $W$ includes candidate $a$ with probability one. Whilst, the only costly voting equilibrium is $\hat{A}_i=\emptyset$ for all $i\in N$ which leads to an election outcome which need not include candidate $a$. To see this observe that in every costly voting equilibrium it must be the case that $\hat{A}_1=\hat{A}_3=\hat{A}_4=\emptyset$, since voters 1, 3, and 4, are indifferent between all election outcomes and hence indifferent between all lotteries. But $n/k=2>1$, and so voter $2$ cannot ensure the election of candidate $a$ and hence will also abstain. Thus, the sincere election outcome need not be attainable as a costly voting equilibrium.
\end{proof}

We now consider the equilibria of the GreedyMonroe rule when a stronger equilibrium concept is applied. Under the Nash equilibrium (and also the costly voting equilibrium) concept, an equilibrium occurs when no voter can strictly profit from a unilateral deviation. This however, does not capture the possibility that voters may engage in coalitions to gain strategic advantage. To capture this idea we consider the notion of \emph{strong Nash equilibria} (introduced in~\cite{Auma59}) and the costly voting version \emph{strong costly voting equilibria}. Both concepts are formally defined below.

Recall that given a ballot profile $\hat{A}$ and voting rule, we denote the lottery of election outcomes produced by the voting rule by $L_{\hat{A}}$. Furthermore, given a set of voters $S\subseteq N$ we denote the reported ballots of voters in $S$ by $\hat{A}_S$; that is, $\hat{A}_{S}=(\hat{A}_{i})_{i\in S}$.

\begin{definition}\emph{[Strong (pure) Nash equilibrium (strong PNE)]}
A reported ballot profile $\hat{A}$ is a strong PNE if for every $S\subseteq N$, there does not exist any set of ballot profiles $\hat{A}_{S}'=(\hat{A}_{i}')_{i\in S}$ such that
\begin{align*}
L_{(\hat{A}_{S}', \hat{A}_{N\backslash S})}&\spref_i L_{\hat{A}}&&\text{for all $i\in S$.}
\end{align*}
\end{definition}

\begin{definition}\emph{[Strong costly voting equilibrium]}
A reported ballot profile $\hat{A}$ is a strong costly voting equilibrium if it is a strong PNE and also a costly voting equilibrium. 
\end{definition}

\begin{remark}
A strong PNE is weakly Pareto efficient~\cite{Auma59}. This is also true for strong costly voting equilibria.
\end{remark}

Considering the strong costly voting equilibrium reduces the number of equilibria admitted by the GreedyMonroe voting rule. For example, the reported profile $\hat{A}=\emptyset$, i.e. where all voters abstain, is a costly voting equilibrium for all preference profiles since $n/k>1$ and so there is no profitable unilateral deviation for any voter. However, under the strong costly voting equilibrium this is no longer the case since coalitions of sizes at least $n/k$ may be able to find strictly profitable deviations. 

The following result shows that every strong costly voting equilibrium is desirable in terms of proportional representation. That is, there are no equilibria which do not satisfy PJR. This is stated and proven in the theorem below. We note however, that the result does not guarantee the existence of such an equilibrium for every preference profile.

The intuition for the result is as follows: If there was a strong costly voting equilibrium which did not satisfy PJR then there would be a group of voters $X$ with $|X|\ge \ell n/k$ (for some positive integer $\ell$) who are represented by strictly less than $\ell$ candidates in $\cup_{i\in X} A_i$, and also unanimously support $\ge \ell$ candidates (i.e. those in $\cap_{i\in X} A_i$). But, under the GreedyMonroe rule, this group can ensure at least $\ell$ candidates in $\cap_{i\in X} A_i$ are elected since they are a group of size $\ge \ell n/k$, and so as a group there is a deviation which strictly benefits every voter in $X$.

%
%
%

\begin{theorem}
Let the voting rule be GreedyMonroe, and assume $k$ divides $n$. Every strong costly voting equilibrium satisfies PJR.
\end{theorem}

\begin{proof}
Suppose for the purpose of a contradiction that $\hat{A}$ is a strong costly voting equilibrium, let $L$ denote the lottery over election outcome under $\hat{A}$, and let $W$ be a some realised outcome.

For the purpose of a contradiction suppose that $\hat{A}$ does not satisfy PJR; that is, there exists a realised outcome, say $W$, which does not satisfy PJR. Thus, there exists a positive integer $\ell$ and a group of voter $X$ with $|X|\ge \ell n/k$ and $|\cap_{i\in X}A_i|\ge \ell$ such that 
$$|W\cap(\cup_{i\in X} A_i)|<\ell.$$
It follows that for all $i\in X$, $|W\cap A_i|<\ell$.

Now consider the deviation such that $\hat{A}_i'=\cap_{i\in X} A_i$ for all $i\in X$. Under this deviation, when all other voters strategies are unchanged,\footnote{In fact this argument holds even if we allow other coalitions to simultaneously deviate} the GreedyMonroe elects at least $\ell$ candidates from $\cap_{i\in X} A_i$ with probability one. To see this note that there are $\ge \ell n/k$ voters unanimously supporting $\ge \ell$ candidates, thus the first-stage of GreedyMonroe can not terminate without $\ge \ell$ candidates in $\cap_{i\in X} A_i$ being elected.

Thus, if $L'$ denotes the lottery under the deviation by group $X$ then every election outcome $W'$ is such that we have $|W'\cap A_i|\ge \ell$ for all $i\in X$. It follows immediately that under the $ul$-extension of preferences every voter strictly prefers the lottery $L'$ over $L$ (the worst-case outcome under $L'$ is strictly better than that under $L$). Thus, every voter in the group $X$ is strictly better off from the deviation, and so $\hat{A}$ cannot be a strong costly voting equilibrium which is a contradiction.
\end{proof}

\subsection{An impossibility result in the implementation of proportional representation}


The previous sections were focused on the existence of robust voting rules which admit at least one equilibrium which satisfy proportional representation. A natural question is whether all equilibria of the voting rule satisfy proportional representation and, if not, does there exist such a voting rule? Within the literature of implementation theory, if such a voting rule were to exist it would be said to \emph{fully implement} proportional representation. A voting rule which fully implements proportional representation is indeed desirable since, without the ability to coordinate voter beliefs about others' actions, there is no guarantee that the desirable equilibria constructed in the previous section will actually be attained. Unfortunately, if there exists such a deterministic voting rule then it must be the case that there exists instances such that a costly voting equilibrium need not exist (in such cases the voting rule may be considered to be `degenerate'). 

We begin by defining an arbitrary voting rule and the relevant parameters; Let $M$ be a finite message space which voters take actions in, and given a positive integer $k$, denote the space of all candidate subsets of $C$ with size $k$ by $\mathcal{W}_k$. Let $\Delta(\mathcal{W}_k)$ denote the set of all lotteries over elements of $\mathcal{W}_k$. Note that the message space $M$ is an arbitrary action space and is not restricted to approval ballots. Given $k$ and $n$ (the pre-determined size of the election outcome and number of voters), a voting rule is a mapping of a sequence of $n$ elements of $M$ to an element of $\Delta(\mathcal{W}_k)$. That is, $V_{(n,k)}: \prod_{i=1}^n M\rightarrow \Delta(\mathcal{W}_k)$.

If a voting rule treats all voters equally then the voting rule is considered to \emph{anonymous}. More formally, a voting rule $V_{(n,k)}$ is anonymous if given any input of `votes', say $(m_1, \ldots, m_m)$, any reordering (or permutation) of the inputs leads to the same outcome, i.e.,  $V_{(n,k)}(m_1, \ldots, m_n)=V_{(n,k)}(m_{\sigma(1)}, \ldots, m_{\sigma(n)})$ for all permutations $\sigma: [n]\rightarrow [n]$.
%

We now prove an impossibility results which shows that there is no resolute, anonymous and deterministic voting rule which fully implements PJR and also guarantees the existence of costly voting equilibria.

The intuition for the proof is as follows; if a voting rule guarantees the existence of equilibria and fully implements PJR then there must exist an equilibrium for every preference profile -- even in the case where PJR is trivially satisfied. Under a deterministic voting rule, when electing two candidates out of three, there is one candidate, say $c$, who is necessarily excluded from the equilibrium $\hat{m}$ election outcome under preference profile $A$. We then consider an alternate but similar preference profile, $A'$ which requires that candidate $c$ is elected for PJR. Thus, if the voting rule fully implements PJR there must exist a deviation from $\hat{m}$ under the preference $A'$. The similarity between the two preference profiles $A$ and $A'$ considered shows that this deviation must have in fact been an available and profitable deviation from $\hat{m}$ in the original equilibrium under profile preferences $A$ -- this contradicts the original equilibrium. 

\begin{theorem}\emph{[Impossibility of (deterministic) PJR implementation]}
Any resolute, anonymous, and deterministic voting rule which always admits at least one costly voting equilibrium, does not fully implement PJR. 
\end{theorem}

\begin{proof}
Let $|N|=4, k=2, C=\{a,b,c\}$. Suppose for the purpose of a contradiction that $V_{(n,k)}$ is a resolute, anonymous, and deterministic voting rule which fully implements PJR in costly voting equilibria, and always admits at least one equilibrium for every preference profile. 

Let the preference profile $A$ be 
\begin{align*}
A_1=\{a\}, A_2=\{b\}, A_3=\{c\},A_4=\emptyset.
\end{align*}
By assumption there exists a costly voting equilibrium, say $\hat{m}\in \prod_{i=1}^4 M$, due to symmetry and without loss of generality assume that the election outcome is $W=\{a,b\}$; that is,  
\begin{align}\label{impossibiltiyequation}
V_{(n,k)}(\hat{m}_1, \hat{m}_2, \hat{m}_3, \hat{m}_4)=\{a,b\}.
\end{align}
Note that this trivially satisfies PJR. Furthermore, given the preference profile, a necessary condition for a costly voting equilibrium to support the outcome $W$ is that $\hat{m}_3=\hat{m}_4=\emptyset.$

This follows because the outcome $W$ is voter $3$'s least preferred outcome, and voter 4 is indifferent between all election outcomes. Thus, (\ref{impossibiltiyequation}) simplifies to 
\begin{align}\label{impossibiltiyequation1}
V_{(n,k)}(\hat{m}_1, \hat{m}_2, \emptyset, \emptyset)=\{a,b\}.
\end{align}

Now consider the alternate preference profile $A'$
\begin{align*}
A_1'=\{a\}, A_2'=\{b\}, A_3'=\{c\},A_4'=\{c\}.
\end{align*}
Clearly, $W$ does not satisfy PJR for the profile $A'$ and so $\hat{m}$ must not be an equilibrium, since $V_{(n,k)}$ is assumed to fully implement JR. In particular, there must exists a voter $i\in N$ such that a unilateral deviation leads to a strictly more preferred election outcome. Since voters $1,2,3$ have the same preferences under $A$ (and $\hat{m}$ is an equilibrium for $A$) there can be no profitable deviation for these voters. Thus, voter 4 must have a strictly profitable deviation, say $\hat{m}_4'\neq \emptyset$, which leads to a new election outcome including candidate $c$; that is, 
$$V_{(n,k)}(\hat{m}_1,\hat{m}_2,\emptyset,\hat{m}_4')=W\qquad \text{such that } c\in W.$$
Furthermore, by the anonymity of the voting rule $V_{(n,k)}$ we also have 
$$V_{(n,k)}(\hat{m}_1,\hat{m}_2,\hat{m}_4', \emptyset)=W\qquad \text{such that } c\in W.$$
But this contradicts the assumption that $\hat{m}$ is an equilibrium for the original preference profile $A$, since if voter $3$ deviates and reports $\hat{m}_3=\hat{m}_4'$ in (\ref{impossibiltiyequation1}) then a strictly preferred election outcome will be produced. We conclude that there is no such voting rule $V_{(n,k)}$ -- this completes the proof.
\end{proof}

\begin{remark}
If voters have worst-case preference extensions, then the above impossibility result follows identically for random voting rules and/or if we allow mixed-strategies.
\end{remark}

We conclude by highlighting that despite the impossibility result, under the costly voting equilibrium concept, it is unclear whether the impossibility result holds for the strong costly voting equilibrium concept.

\section{Conclusion and outlook}

In this paper we proposed an alternate equilibrium approach for achieving PR in the presence of strategic voters. Within the setting of voters with dichotomous preferences over candidates, we designed a \emph{strategically robust} PR-based voting rule, called GreedyMonroe, which is not strategyproof but ensures the existence of (possibly insincere) costly voting equilibria satisfying PR with respect to voters' private and unrevealed preferences. In contrast, the commonly applied AV-rule was shown not to be strategically robust. Furthermore, we proved a limit on `how strategically robust' a PR-based voting rule can be; we showed that there is no PR-based voting rule\footnote{We assumed the voting rule is deterministic, resolute and anonymous.} which ensures that \emph{every} costly voting equilibrium satisfies PR. Collectively, these results highlight the possibility and limit of achieving PR in the presence of strategic voters and a positive role for mechanisms, such as pre-election polls, which coordinate voter behaviour towards equilibria which satisfy PR.

There still remains many natural and unanswered questions in this line of research such as whether best response dynamics converge to equilibria, extending equilibria results to more general preference domains, and extending our impossibility result to random mechanisms.

\bibliographystyle{ACM-Reference-Format}
\bibliography{abb_bl,adt_bl}

\appendix
%
%
%

\begin{acks}
	I would like to thank Haris Aziz, Gabriele Gratton, Richard Holden, and Carlos Pimienta for their helpful feedback throughout the drafting stages. 
	I also acknowledge the support provided to me by the UNSW Scientia PhD fellowship and the Data61 top-up PhD scholarship.
\end{acks}

	\section{Appendix}

\subsection{Omitted proofs Section 3}

\begin{proof}[Proof of Theorem~\ref{theorem: GreedyJR not strategyproof}]
We prove the claim explicitly via a counter-example. Let $|N|=6$, $C=\{a, b,c,d \}$ and $k=3$ with
\begin{align*}
A_1&= \{c\}\\
A_2&= \{c,b\}\\
A_3, A_4&= \{b\}\\
A_5, A_6&= \{d\}.
\end{align*}
For the purpose of a contradiction suppose that the voting rule is strategyproof; that is, $\hat{A}_i=A_i$ must be a best response for all voters. Under the sincere ballot profile $\hat{A}=A$, from voter $2$'s perspective the worst-case (non-zero probability) election outcome is $W=\{b,d,a\}$ which includes just one of voter $2$'s (truthfully) approved candidates, i.e., $|W\cap A_2|=1$. This occur if in the first iteration of the first-stage candidate $b$ is elected and then the two voters selected to have their ballot removed includes voter $2$ -- in this case, candidate $c$ will no longer have sufficient support to guarantee election.

Now, if voter $2$ deviates and instead reports the insincere ballot $\hat{A}_2'=\{c\}$ then with probability one the election outcome is $W'=\{b,c,d\}$ which includes two of voters $2$'s (truthfully) approved candidates. 

Thus if we denote the lottery under sincere voting by $L$ and the lottery under voter $2$'s unilateral deviation by $L'$, we have $\pi_{L'}^{(2)}(0)=\pi_{L}^{(2)}(0)=0$ and $\pi_{L'}^{(2)}(1)=0<\pi_{L}^{(2)}(1)$, thus $L'\spref_2 L$. That is, voter 2 has a strictly profitable deviation from $\hat{A}_2=A_2$, and hence the GreedyMonroe voting rule is not strategyproof.
\end{proof}

\subsection{Omitted proofs Section 3.1}

The AV-rule is an approval-based election rule which selects the $k$ candidates with the highest approval scores with respect to the reported ballots $\hat{A}$, i.e. the approval score of a candidate $c$ with respect to $\hat{A}$ is $s(c, \hat{A})=|\{i\in N\, :\, c\in \hat{A}_i\}|$. In case of ties, a uniform at random tie-breaking rule is assumed to apply.\footnote{Other tie-breaking rules can be applied here (deterministic or random) without affecting the properties considered in this section.} 

\begin{proposition}\label{proposition: AV-lex no costly voting equil}
Under the AV voting rule, there need not be any costly voting equilibrium supporting a PJR outcome. 
\end{proposition}

\begin{proof}
Let $|N|=4, k=2, C=\{a,b,c\}$ and the preference profile be given as follows:
\begin{align*}
A_1, A_2&=\{a\}\\
A_3&=\{b,c\}\\
A_4&=\emptyset.
\end{align*}
For the purpose of a contradiction suppose that there exists a reported ballot profile $\hat{A}$ such that PJR is satisfied in every non-zero probability election outcome; that is, candidate $a$ is elected with probability one. 

First note $A_4=\emptyset$, and so it must be that $\hat{A}_4=\emptyset$ since $\hat{A}$ is a costly voting equilibrium. Furthermore, notice that since candidate $a$ is elected with probability one in equilibrium, voter 3 must abstain from voting since for every (non-zero probability) election outcome in this equilibrium, say $W$, we have $|W \cap A_3|=1$; that is, $\pi_L^{(3)}(1)=1$ where $L$ is the equilibrium induced lottery. But by the pigeon-hole principle \emph{every} election outcome $W'$ of size $k=2$ guarantees $|W'\cap A_3|\ge 1$ -- and so voter 3 can do no worse by abstaining. We conclude that if PJR is satisfied in a costly voting equilibrium then $\hat{A}_3=\hat{A}_4=\emptyset$. 

It now follows that the approval scores are such that 
$$ \min\{s(b, \hat{A}),s(c, \hat{A})\}+1< s(a, \hat{A}),$$
otherwise, voter $3$ could deviate from $\hat{A}_3=\emptyset$ and submit $\hat{A}_3'=\{b,c\}$ -- this would induce a new lottery $L'$ such that $\pi_{L'}^{(3)}(1)<1=\pi_L^{(3)}(1)$, and hence $L'\spref_3 L$ which would contradict the equilibrium $\hat{A}$. In addition, $\min\{s(b, \hat{A}),s(c, \hat{A})\}\ge 0$ and $s(a, \hat{A})\le 2$ since at most 2 voters submit non-empty ballots. Thus,
\begin{align*}
0\le &\min\{s(b, \hat{A}),s(c, \hat{A})\}+1< s(a, \hat{A})\le 2,
\end{align*}
which implies, due to the integer-valued $s(\cdot, \cdot)$ function, that
\begin{align*}
\min\{s(b, \hat{A}),s(c, \hat{A})\}=0\quad \text{ and } \quad s(a, \hat{A})= 2.
\end{align*}
We immediately infer that $a\in \hat{A}_1, \hat{A}_2$. Without loss of generality assume $s(c, \hat{A})=0$. Now suppose voter $1$ deviates and abstains from voting, i.e, $\hat{A}_1'=\emptyset$, leading to the new (off-equilibrium) ballot profile $\hat{A}=(\hat{A}_1', \hat{A}_{-1})$ and induced lottery $L'$. Under this lottery, since $s(c,\hat{A}')=0$ and $s(a, \hat{A}')=1$, candidate $a$ is elected with probability one under $L'$ and so voter $1$ is indifferent between the ballots $\hat{A}_1\neq \emptyset$ and $\hat{A}_1'=\emptyset$. Thus, voter 1 will optimally deviate from $\hat{A}$ -- this contradicts the assumption that $\hat{A}$ is a costly voting equilibrium. We conclude that under the AV-rule, in this instance there is no PJR outcome supported by a costly voting equilibrium.
\end{proof}

%
%
%
%

%
%
%
%

\end{document}